\documentclass[12pt]{article}
\usepackage{amsmath,amssymb,amsthm}
\usepackage{booktabs}
\usepackage{geometry}
\usepackage[utf8]{inputenc}
\usepackage[super,numbers,sort&compress]{natbib}
\usepackage{url}
\newcommand{\pval}{\ensuremath{p}}
\newcommand{\nb}{\ensuremath{nb}}
\usepackage{hyperref}
\hypersetup{hidelinks}
\newcommand{\Neutral}{\ensuremath{\mathcal{N}}}
\newtheorem{definition}{Definition}
\newtheorem{proposition}{Proposition}
\newtheorem{theorem}{Theorem}

\newtheorem{remark}{Remark}
\begin{document}
\title{The Neutrality Boundary Framework: Quantifying Statistical Robustness Geometrically}
\author{Thomas F. Heston\\
Department of Family Medicine, University of Washington,\\
1959 NE Pacific Street, Box 356390,\\
Seattle, WA 98195, USA}
\maketitle
\begin{abstract}
We introduce the Neutrality Boundary Framework (NBF), a set of geometric metrics for quantifying statistical robustness and fragility as the normalized distance from the neutrality boundary—the manifold where effect equals zero. The neutrality boundary value $\nb \in [0,1)$ provides a threshold-free, sample-size invariant measure of stability complementing traditional effect sizes and \pval-values. We derive the general form $\nb = |\Delta-\Delta_0|/(|\Delta-\Delta_0|+S)$, where $S>0$ is a scale parameter for normalization; prove boundedness and monotonicity, and provide domain-specific implementations: Risk Quotient (binary outcomes), partial $\eta^2$ (ANOVA), and Fisher z-based measures (correlation). Unlike threshold-dependent fragility indices, NBF quantifies robustness geometrically across arbitrary significance levels and statistical contexts.
\end{abstract}
\section{Introduction}
Quantifying the robustness of statistical findings remains fundamental for inference. The reproducibility crisis in the biomedical and social sciences \cite{ioannidis_why_2005} has intensified the scrutiny of statistical stability. Existing fragility metrics (Fragility Index, Fragility Quotient) identify minimum outcome changes needed to cross $p=0.05$ thresholds \cite{walsh_statistical_2014, ahmed_does_2016, kampman_statistical_2022, gonzalez-del-hoyo_fragility_2023}. The robustness index extends this framework by examining changes in sample size while maintaining distributional integrity \cite{heston_robustness_2023}. Although widely applied, these indices remain threshold-dependent and probabilistic in nature. Effect-size reporting frameworks \cite{cohen_statistical_1988, lakens_2013, sullivan_2012} improve interpretability, yet do not quantify robustness itself.
We introduce a geometric alternative. The Neutrality Boundary Framework (NBF) quantifies robustness as the normalized distance from the neutrality boundary \Neutral—the set of parameter values where the effect equals zero. The resulting measure $\nb \in [0,1)$ is continuous and threshold-free, complementing both traditional effect sizes \cite{cohen_statistical_1988} and significance testing conventions \cite{wasserstein_asa_2016}. This geometric framework focuses on stability rather than magnitude, providing a unified robustness metric across statistical contexts.
Our contributions: (1) We establish the general NBF with formal proofs of boundedness, monotonicity, and sample-size invariance (Section~2). (2) We derive domain-specific implementations for binary, multi-group, and correlation contexts (Section~3). (3) We demonstrate NBF as a unified robustness measure independent of significance testing conventions.
\section{General Framework}
\begin{definition}[Neutrality Boundary Framework]
\label{def:boundary}
Let $\Theta$ be a parameter space and $T(\mathbf{X})$ a statistic. The \emph{neutrality boundary framework} \Neutral\ is the set of parameter values for which the scientific effect is null: risk ratio $=1$, risk difference $=0$, $r=0$ (zero linear association), etc.
\end{definition}
\begin{definition}[Neutrality Boundary]
\label{def:nbf}
For a contrast with observed value $\Delta$ and neutrality $\Delta_0 \in \mathcal{N}$, the \emph{neutrality boundary value} is
\begin{equation}
\label{eq:nb}
\nb = \frac{|\Delta - \Delta_0|}{|\Delta - \Delta_0| + S} \in [0,1),
\end{equation}
where $S>0$ is a scale parameter ensuring normalization.
\end{definition}
\begin{proposition}[Boundedness and monotonicity]
\label{prop:bounds}
For any $S>0$ constant with respect to $\Delta$, we have $\nb \in [0,1)$ with $\nb=0$ iff $\Delta=\Delta_0$, and $\nb$ increases monotonically in $|\Delta-\Delta_0|$.
\end{proposition}
\begin{proof}
Since $|\Delta-\Delta_0| \geq 0$ and $S>0$, the denominator exceeds the numerator, so $\nb < 1$. Lower bound: $\nb \geq 0$ with equality iff $\Delta=\Delta_0$. Monotonicity: taking the derivative with respect to $|\Delta-\Delta_0|$ (treating $S$ as constant),
\[
\frac{\partial \nb}{\partial |\Delta-\Delta_0|} = \frac{S}{(|\Delta-\Delta_0|+S)^2} > 0
\]
for all $S>0$.
\end{proof}
\begin{proposition}[Sample-size invariance]
\label{prop:invariance}
If $\Delta$ and $S$ represent effect magnitude and intrinsic dispersion (not standard error), then the population $\nb$ is invariant to sample size $n$. For estimators, $\mathbb{E}[\hat{\nb}]$ is asymptotically invariant when $\hat{\Delta}$ and $\hat{S}$ are consistent for population quantities.
\end{proposition}
\begin{proof}
Standard error scales as $\mathrm{SE} \propto n^{-1/2}$. If $\Delta$ and $S$ measure population-level quantities (e.g., mean difference and pooled SD), both are $O(1)$ in $n$. Thus the population $\nb = f(\Delta, S)$ with $f$ independent of $n$. Finite-sample estimators $\hat{\nb}$ have $n$-dependent variance, but $\mathbb{E}[\hat{\nb}] \to \nb$ as $n \to \infty$ by consistency of $\hat{\Delta}$ and $\hat{S}$.
\end{proof}
\begin{proposition}[Valid scale parameters]
\label{prop:scale}
A scale parameter $S$ in \eqref{eq:nb} is valid if: (i) $S>0$ for all data, (ii) $\nb$ is invariant to affine transformations, (iii) $S$ reflects natural dispersion (e.g., pooled variance, maximum independence variance).
\end{proposition}
\begin{remark}
Proposition~\ref{prop:scale} ensures $\nb$ measures intrinsic stability rather than arbitrary scaling. Domain-specific choices of $S$ are determined by the contrast's natural variability structure.
\end{remark}
\section{Domain-Specific Implementations}
\subsection{Categorical outcomes: Risk Quotient}
The Risk Quotient defines the geometric distance from independence in contingency tables. This concept is related to the \emph{Unit Fragility Index (UFI)} \cite{feinstein_unit_1990}, a one-unit geometric perturbation of the difference in proportions under fixed margins. For $2\times 2$ contingency tables under fixed margins, a one-unit outcome exchange between cells theoretically shifts $|ad-bc|$ by $N$, corresponding to a lattice spacing of $4/N$ in the derived $RQ_{2\times 2} = 4|ad-bc|/n^2$ scale. In balanced designs this increment equals Feinstein's UFI. In practice, $RQ$ is a fixed table quantity; the $4/N$ relation represents the minimum discrete spacing of possible $RQ$ values, not an empirical increment.
The UFI was subsequently connected to probability-based significance reversal, demonstrating how a single data change could transition a contrast across a $p$-value threshold \cite{walter_statistical_1991}. These works illustrate the bridge between geometric and probabilistic conceptions of robustness. Neutrality Boundary Framework (NBF) extends this lineage by formalizing fragility and robustness in geometric terms—normalized distance from neutrality—while remaining complementary to probability-based inference.
\paragraph{Canonical two-arm case.}
For a $2\times 2$ table with cells $(a,b;c,d)$ and total $n=a+b+c+d$, the Risk Quotient is
\begin{equation}
\mathrm{RQ} = \frac{4|ad-bc|}{n^2} \in [0,1].
\end{equation}
In canonical NBF form,
\begin{equation}
\nb = \frac{\mathrm{RQ}}{1+\mathrm{RQ}} = \frac{|ad-bc|}{|ad-bc| + n^2/4} \in [0,1),
\end{equation}
which matches $\frac{|\Delta-\Delta_0|}{|\Delta-\Delta_0|+S}$ with $\Delta=|ad-bc|$, $\Delta_0=0$, $S=n^2/4$. The $2\times 2$ form yields a discrete geometric lattice where one unit change corresponds to a step of $4/N$ in $RQ$ space—aligning with Feinstein’s UFI in balanced designs.
\paragraph{Generalization to $r\times c$ tables.}
For higher-dimensional contingency tables, $RQ$ generalizes as a normalized absolute deviation from independence:
\begin{equation}
RQ = \frac{1}{n}\sum_{i,j}\bigl|O_{ij}-E_{ij}\bigr|, \qquad E_{ij} = \frac{(\mathrm{row}_i\ \mathrm{total})(\mathrm{col}_j\ \mathrm{total})}{n}.
\end{equation}
This maintains boundedness ($0\le RQ\le 1$) and a geometric interpretation as average per-cell distance from neutrality (independence manifold $O_{ij}=E_{ij}$). Unit-step increments depend on the marginal structure and generally do not reduce to $4/N$. Applying the canonical transformation $\nb = RQ/(1+RQ)$ preserves normalization and monotonicity across multi-category designs.
\subsection{Multi-group ANOVA}
For fixed-effect one-way ANOVA with between-group df $\mathrm{df}_b$, within-group df $\mathrm{df}_w$, and F-statistic:
\begin{equation}
\nb = \eta^2_{\mathrm{partial}} = \frac{\mathrm{df}_b \cdot F}{\mathrm{df}_b \cdot F + \mathrm{df}_w} \in [0,1).
\end{equation}
Partial eta-squared provides a natural 0--1 bounded measure of effect magnitude in ANOVA contexts \cite{richardson_eta_2011}. Alternatively, a monotone 0--1 transformation is $\nb = f/(1+f)$ where $f = \sqrt{\eta^2/(1-\eta^2)}$ is Cohen's $f$ \cite{cohen_statistical_1988}. For one-way designs, $\eta^2 = \eta^2_{\text{partial}}$, so both forms are monotone-equivalent. Neutrality: all group means equal.
\subsection{Correlation: Distance to Independence}
For Pearson correlation $r \in (-1,1)$ with Fisher z-transformation \cite{fisher_frequency_1915} $z = \operatorname{atanh}(r)$:
\begin{equation}
\nb = \mathrm{DTI} = \frac{|z|}{1+|z|} \in [0,1).
\end{equation}
The Fisher z-transform stabilizes the variance of the correlation coefficient and provides an unbounded scale suitable for the canonical NBF form. Neutrality: $r=0$ (zero linear association).
\begin{theorem}[Unified form]
\label{thm:unified}
Each domain-specific implementation satisfies Definition~\ref{def:nbf} with appropriate choice of $\Delta$, $\Delta_0$, and $S$ meeting Proposition~\ref{prop:scale}.
\end{theorem}
\begin{proof}[Proof sketch]
Binary ($2 \times 2$): $\Delta = |ad-bc|$, $\Delta_0=0$, $S=n^2/4$, yielding $\nb = \mathrm{RQ}/(1+\mathrm{RQ})$. ANOVA: Transform of F-statistic via $\eta^2_{\text{partial}}$, which can be expressed as $\nb = \mathrm{df}_b \cdot F/(\mathrm{df}_b \cdot F + \mathrm{df}_w)$. Correlation: $\Delta=|z|$ (Fisher z-transform), $\Delta_0=0$, $S=1$, yielding $\nb=|z|/(1+|z|)$. Each satisfies boundedness and monotonicity by Proposition~\ref{prop:bounds}.
\end{proof}
\section{Properties and Interpretation}
\begin{table}[ht]
\centering
\caption{Interpretation of neutrality boundary values}
\label{tab:interpretation}
\begin{tabular}{@{}lll@{}}
\toprule
$\nb$ range & Interpretation & Meaning \\
\midrule
$0$--$0.05$ & Extremely fragile & Near neutrality \\
$0.05$--$0.10$ & Fragile & Slight separation \\
$0.10$--$0.25$ & Moderately robust & Stable separation \\
$0.25$--$0.50$ & Robust & Strong separation \\
$>0.50$ & Very robust & Far from neutrality \\
\bottomrule
\end{tabular}
\end{table}
Unlike \pval-values (which quantify $\Pr(\text{data} \mid H_0)$ and are subject to widespread misinterpretation \cite{wasserstein_asa_2016}), $\nb$ measures geometric distance from neutrality. Unlike confidence intervals (constructed from $\mathrm{SE} \propto n^{-1/2}$), $\nb$ is population-scale invariant; estimator variance decreases with $n$ (Proposition~\ref{prop:invariance}). A finding with $\nb=0.42$ indicates the same geometric stability whether $n=30$ or $n=3000$, though inferential certainty increases with $n$. The geometric robustness quantified by $\nb$ aligns with recent critiques of statistical fragility and reproducibility in biomedical research \cite{kampman_statistical_2022, gonzalez-del-hoyo_fragility_2023, ho_fragility_2022}.
\paragraph{Complementary to probability.}
Geometric measures ($\nb$, effect size) describe data properties; probabilistic measures (\pval-values, CI) describe sampling uncertainty. The distinction parallels the ASA's call for moving beyond binary interpretations of statistical significance \cite{wasserstein_asa_2016}. Low \pval-value with high $\nb$ indicates significant and stable separation. Low \pval-value with low $\nb$ reflects statistical significance without geometric robustness (common in large samples). High \pval-value with high $\nb$ suggests underpowered detection of genuine separation. Both perspectives are necessary for complete inference.
\section{Discussion}
NBF provides a threshold-free geometric alternative to probability-based fragility metrics. Unlike data-manipulation approaches (FI, FQ) or sample-size scaling methods (robustness index) \cite{heston_robustness_2023, heston_redefining_2024}, NBF maintains both distributional integrity and sample-size invariance through geometric normalization. By normalizing distance from neutrality, $\nb \in [0,1)$ enables direct comparison across contexts while maintaining sample-size invariance. The framework assumes well-defined neutrality boundaries and appropriate scale parameters. Distributional assumptions (normality, independence) apply as in standard inference.
\paragraph{Extensions.}
Future work may address: (i) asymptotic distributions of $\nb$ estimators, (ii) confidence intervals for $\nb$, (iii) extensions to survival analysis, longitudinal models, and continuous two-sample comparisons, (iv) connections to information-geometric distances.
\section*{Declarations}
\paragraph{Competing interests} The author declares no competing interests.
\paragraph{Funding} No funding was received for this work.
\paragraph{Author contributions} The author conceived, developed, and wrote the manuscript; all derivations were verified by the author, who takes full responsibility for the content.
\paragraph{Data availability} No data were used.
\bibliographystyle{vancouver}
\bibliography{refs}
\end{document}